\documentclass[11pt]{article}
\usepackage{graphicx,amsmath,amssymb,amsthm,subfig,url,hyperref}
\usepackage{tikz-cd} 
\usepackage{bbm,float}
\usepackage[titletoc]{appendix}

\newtheorem{theorem}{Theorem}
\newtheorem{lemma}[theorem]{Lemma}
\newtheorem{corollary}[theorem]{Corollary}

\numberwithin{theorem}{section} % important bit
\mathchardef\mhyphen="2D

\usepackage[top=1in, bottom=1in, left=1in, right=1in]{geometry}

\newtheorem*{theorem*}{Theorem}

\graphicspath{{figures/}}
\author{Nancy Lynch and Mien Brabeeba Wang}

\title{Integrating Temporal Information to Spatial Information in a Neural Circuit}

\begin{document}
\maketitle
\begin{abstract}
In this paper, we consider networks of deterministic spiking neurons, firing synchronously at discrete times; such spiking neural networks are inspired by networks of neurons and synapses that occur in brains. We consider the problem of translating temporal information into spatial information in such networks, an important task that is carried out by actual brains.

Specifically, we define two problems:  ``First Consecutive Spikes Counting (FCSC)'' and ``Total Spikes Counting (TSC)'', which model spike and rate coding aspects of translating temporal information into spatial information respectively. Assuming an upper bound of $T$ on the length of the temporal input signal, we design two networks that solve these two problems, each using $O(\log T)$ neurons and terminating in time $1$. We also prove that there is no network with less than $T$ neurons that solves either question in time $0$.
\end{abstract}

\section{Introduction}
One of the most important questions in neuroscience is how humans integrate information over time. Sensory inputs such as visual and auditory stimuli are inherently temporal; yet brains can integrate the temporal information into a single concept, such as recognizing a moving object in a visual scene or forming an entity in a sentence. In the above examples, the temporal information spans over a time scale of $1$-$10$ seconds. However, individual neurons only have transient activities with the time scale of $10$-$100ms$. It is not clear how neurons with transient components can process temporal information over a long time range. In this paper, we are going to present a static network to process temporal information and translate it into spatial information with transient components.  

There are two kinds of neuronal codings: {\it rate coding} and {\it temporal coding}. Rate coding is a neural coding scheme assuming most of the information is coded in the firing rate of the neurons. It is most commonly seen in muscle when the higher firing rates of motor neurons correspond to higher intensity in muscle contraction \cite{Adrian1926}. On the other hand, rate coding cannot be the only neural coding brains employ. A fly is known to react to new stimuli and change its direction of flight within $30$-$40$ ms. For a neuron that spikes at around $50Hz$, which is much higher than the average spiking rate, there is only time to produce $1$-$2$ spikes within this window. There is simply not enough time for neurons to decode rate coding accurately \cite{Rieke1996}. Therefore, neuroscientists proposed the idea of temporal coding, assuming the information is coded in the temporal firing patterns. One of the popular temporal codings is the first-to-spike coding, in which the information is encoded in the duration between the stimulus onset and the first spike. By plotting the timing of the first spike in retina ganglion cells, one can recover an approximately accurate image on a retina \cite{Gollisch2008}. 

We propose two toy problems to model how brains extract temporal information from different coding with transient components. {\it``First consecutive spikes counting" (FCSC)} counts the first consecutive interval of spikes, which is equivalent to counting the distance between the first two spikes, a prevalent temporal coding scheme in the sensory cortex. {\it``Total spikes counting" (TSC)} counts the number of the spikes over an arbitrary interval, which is an example of rate coding. To model the transient components of neurons, we consider a memoryless synchronous spiking neuron model where the firing of a neuron only depends on the spike events one time step ago. 

In this paper, we design two networks that solve the above two problems by translating temporal information into spatial information in time $1$ with $O(\log T)$ neurons. We further show that any network with less than $T$ neurons cannot solve the problems in time $0$. It should be noted that Hitron and Parter also considered the TSC problem \cite{Merav2019} with the time bound $O(\log T)$. In this context, we improve the time bound on the TSC problem from $O(\log T)$ to $1$ by carefully updating all digits in binary representation at once instead of sequentially. We would like to remark that although our problems are biologically inspired, the optimal solutions we propose are not biologically plausible. The networks are not noise-tolerant, whereas the neuronal dynamics are highly noisy and it is hard to conceive that the brain uses binary representation as a neuronal representation. However, the analysis serves as a proof of concept that the brain can process temporal information over a long time range using transient components.

The organization of the rest of the section is as follows. In~\autoref{sec: model}, we formally define the spiking neuron model we are working in. In~\autoref{sec: problem}, we define the two biologically-inspired problems ``First Consecutive Spikes Counting" and ``Total Spikes Counting" which correspond to temporal coding and rate coding respectively. In~\autoref{sec: main theorem}, we provide our main results, solving the two problems optimally in both time and the number of the neurons and showing that we cannot do better.
\subsection{Model}\label{sec: model}
In this work, to model the transient aspect of the neurons, we consider a network of memoryless spiking neurons with deterministic synchronous firing at discrete times. Formally, a neuron $z$ consists of the following data with $t\geq 1$
\[
z^{(t)} = \Theta(\sum_{y\in P_z}w_{yz}y^{(t-1)} - b_z)
\]
where $z^{(t)}$ is the indicator function of neuron $z$ firing at time $t$. $b_z$ is the threshold (bias) of neuron $z$. $P_z$ is the set of presynaptic neurons of $z$, $w_{yz}$ is the strength of connection from neuron $y$ to neuron $z$ and $\Theta$ is a nonlinear function. Here we take $\Theta$ as the Heaviside function given by $\Theta(x) = 1$ if $x>0$ and $0$ otherwise. At $t=0$, we let $z^{(0)} = 0$ if $z$ is not one of the input neurons.

For the rest of the paper, we fix an input neuron $x$ and $m$ output neurons $\lbrace y_i\rbrace_{0\leq i < m}$ in a network.

\subsection{Problem Statement}\label{sec: problem}
\subsubsection{First Consecutive Spikes Counting(T) (FCSC(T))}
Given an input neuron $x$ and the max input length $T$, we consider any input firing sequence such that for all $t \geq T,$ $x^{(t)} = 0$. Define $L_x$ in terms of this firing sequence as follows: if $x^{(t)} = 1$ for some $t$, then there must exist integers $\hat{t}, L$ such that for all $t, t<\hat{t}$ we have $x^{(t)} = 0$, for all $i, 0\leq i<L$ we have $x^{(\hat{t} + i)} = 1$, and $x^{(\hat{t} + L)} = 0$. Define $L_x = L$. (i.e., $L$ is the length of the first consecutive spikes interval in the sequence.)  Otherwise, that is if for all $t\geq 0$, $x^{(t)} = 0$, then define $L_x = 0$. 

Then we say a network of neurons solves FCSC(T) in time $t'$ with $m'$ neurons if there exists an injective function $F: \lbrace 0, \dotsb, T\rbrace\rightarrow \lbrace 0, 1\rbrace^m$ such that for all $x$ and for all $t, t\geq T + t'$ we have $y^{(t)} = F(L_x)$ and the network has $m'$ total neurons.

Intuitively, FCSC serves as a toy model for encoding distance between spikes, a prevalent spike coding in the sensory cortex. For mathematical convenience, we model the problem as counting the distance between non-spikes which is mathematically equivalent as counting the distance between spikes in our model.

\subsubsection{Total Spikes Counting(T) (TSC(T))}
Given an input neuron $x$ and the max input length $T$, we consider any input firing sequence such that for all $t \geq T$, $x^{(t)} = 0$. Define $L_x = |\lbrace t: x^{(t)} = 1, 0\leq t < T \rbrace|$ as the total number of spikes in the sequence. Then we say a network of neurons solves TSC(T) in time $t'$ with $m'$ neurons if there exists an injective function $F: \lbrace 0, \dotsb, n\rbrace\rightarrow \lbrace 0, 1\rbrace^m$ such that for all $x$ and for all $t, t\geq T+ t'$ we have $y^{(t)} = F(L_x)$ and the network has $m'$ total neurons.

Intuitively, TSC serves as a toy model for rate coding implemented by spiking neural networks because the network can extract the rate information by counting the number of spikes over arbitrary intervals. 

Notice that in both problems above, a network solves a task in time $t'$ if, for all $t\geq T + t'$ and for all inputs with max length $T$, the network outputs the solution of the task at time $t$. The definition is equivalent to Maass's time complexity for spiking neurons~\cite{Maass1996}. This definition of the time bound makes natural sense since given a max input length of $T$, it is unreasonable to count the time before the end of the input.

\subsection{Main Theorems}\label{sec: main theorem}

Our contributions in this work are to design networks that solve these two problems respectively with matching lower bounds in numbers of neurons.
\begin{theorem}
\label{FCS}
There exists a network that solves FCSC(T) problem with $O(\log T)$ neurons in time $1$.  
\end{theorem}
\begin{theorem}
\label{TSCMain}
There exists a network that solves TSC(T) problem with $O(\log T)$ neurons in time $1$.   
\end{theorem}
It is easy to see that we also have the corresponding information-theoretical lower bound on the number of neurons all being $\Omega(\log T)$ by the requirements of the tasks.

In terms of time bound, we also show that our networks are optimal for FCSC and TSC problem in the following sense:
\begin{theorem}
\label{FCSClower}
There does not exists a network with less than $T$ neurons that solves FCSC(t) problem in time $0$ for all $0\leq t\leq T$.
\end{theorem}
\begin{theorem}
\label{TSClower}
There does not exists a network with less than $T$ neurons that solves TSC(t) problem in time $0$ for all $0\leq t\leq T$.
\end{theorem}

\section{First Consecutive Spikes Counting}

We present the constructions in two stages. At the first stage, we count consecutive spikes in binary transiently. At the second stage, we transform the transient firing into persistent firing. By composing the two stages, we get our desired network.

\subsection{First Stage: Counter Network}  The network contains neurons $z_0, \dotsb, z_n$, $in_1, \dotsb, in_n$ and we build the network inductively. To construct mod $2$ Base Network which counts mod $2$, we have 
\[w_{xz_0} = 1, w_{z_0z_0} = -1, b_{z_0} = 0.5.\]
\begin{figure}
\includegraphics[width=3.3cm]{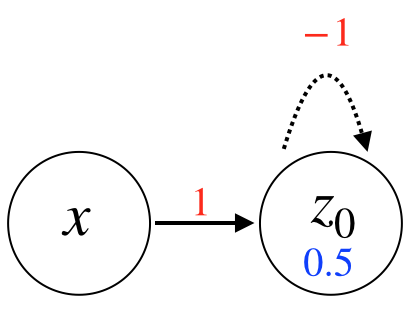}
\centering
\caption{mod $2$ Base Network}
\end{figure}
By noticing that for $t\geq 1$, $z_0^{(t)} = 1$ if and only if $x^{(t-1)} = 1$ and $z_0^{(t-1)} = 0$, we have the following lemma
\begin{lemma}
\label{base}
For the mod $2$ base network, given $t\geq 0$ if for all $t' such that 0\leq t' \leq t$ we have $x^{(t')} = 1$, then at time $t$, $z_0^{(t)} = t \bmod 2$. 
\end{lemma}
Now we iteratively build the network where $1\leq i \leq n$ on top of the mod $2$ base network with the following rule:
\[
w_{xz_i} = i + 1,\ w_{z_jz_i} = 1,\ \forall j, 0\leq j<i, w_{z_kin_i} = 1, \forall k, 0<k\leq i,w_{in_iz_i} = -i-1, w_{z_iz_i} = i
\]
\[
b_{z_i} =  2i + 0.5,\ b_{in_i} = i - 0.5.
\]
\begin{figure}
\includegraphics[width=9.8cm]{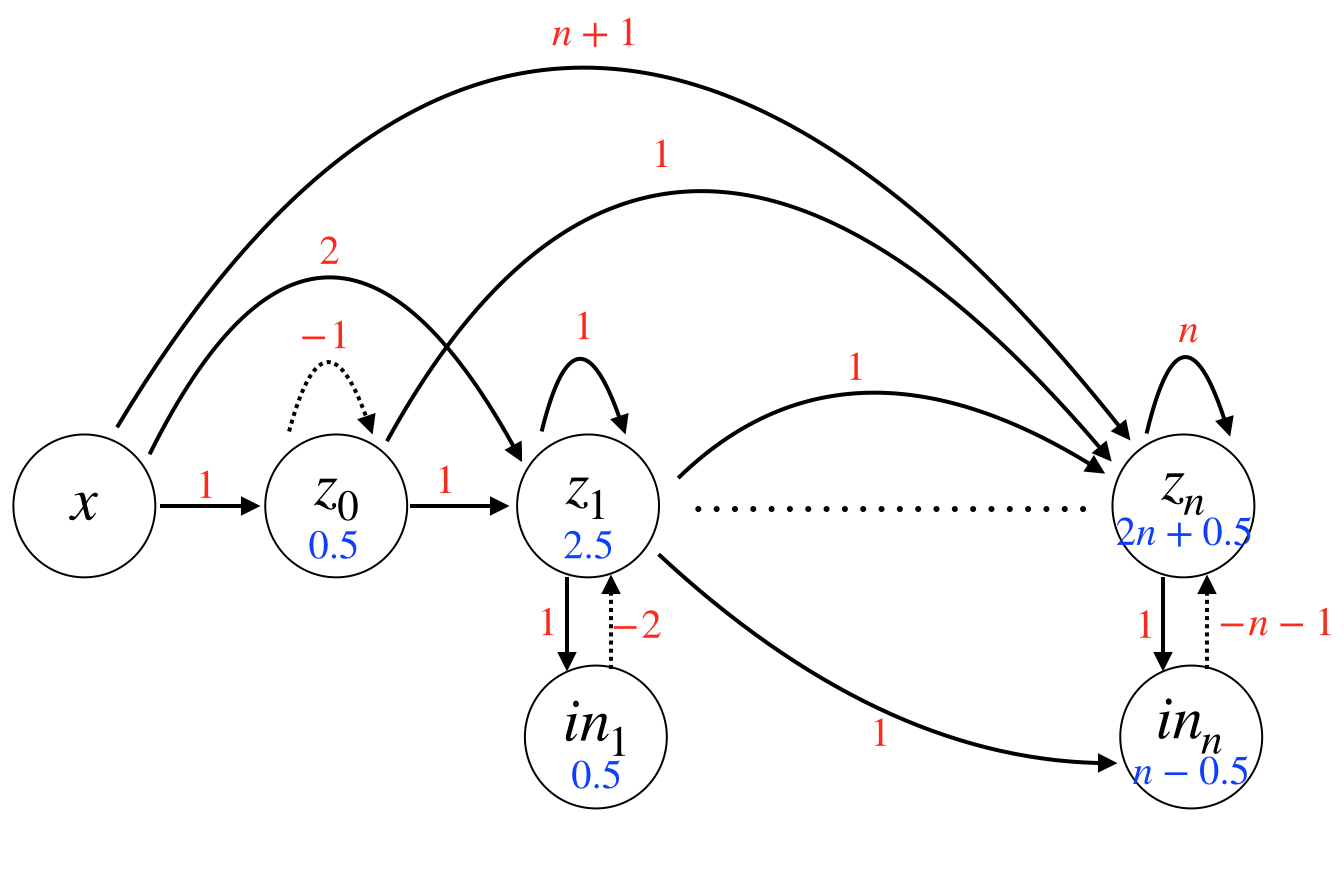}
\centering
\caption{First Stage}
\end{figure}
This completes the construction. From the construction, we can deduce the following lemma.
\begin{lemma}
\label{firing1}
For $i > 0$, neurons $z_i, in_i$ fire according to the following rules: 
\begin{enumerate}
    \item $z_i^{(t)} =1$ if and only if $x^{(t-1)} = 1$, $in_i^{(t-1)} = 0$, and (either for all $j,  0\leq j<i$ we have $z_j^{(t-1)} = 1$ or $z_i^{(t-1)} = 1$).
    \item $in_i^{(t)} = 1$ if and only if for all $j, 1\leq j\leq i$ we have $z_j^{(t-1)} = 1$.
\end{enumerate}
\end{lemma}
\begin{proof}
\textbf{Case (1):}
The potential of $z_i^{(t)}$ is
\begin{multline*}
w_{xz_i}x^{(t-1)} + \sum_{j=0}^{i-1}w_{z_jz_i}z_j^{(t-1)} + w_{in_iz_i}in_i^{(t-1)} + w_{z_iz_i}z_i^{(t-1)} \\= (i+1)x^{(t-1)} + \sum_{j=0}^{i-1}z_j^{(t-1)} -(i+1)in_i^{(t-1)} + iz_i^{(t-1)}\,.    
\end{multline*}
\textbf{Only if:}
Let's show the only if direction for the firing rule of $z_i^{(t)}$ by proving the contrapositive.\\
If $x^{(t-1)} = 0$, then the potential of $z_i^{(t)}$ is
\[
\sum_{j=0}^{i-1}x_j^{(t-1)} -(i+1)in_i^{(t-1)} + iz_i^{(t-1)}\leq 2i < 2i+0.5 = b_{z_i}.
\] 
If $in_i^{(t-1)} = 1$, then the potential of $z_i^{(t)}$ is
\[
(i+1)x^{(t-1)} + \sum_{j=0}^{i-1}z_j^{(t-1)} -(i+1)+ iz_i^{(t-1)}\leq 2i < 2i+0.5 = b_{z_i}.
\]
If there exists $\hat{j}, 0\leq \hat{j}<i$ such that $z_{\hat{j}}^{(t-1)} = 0$ and $z_i^{(t-1)} = 0$, then the potential of $z_i^{(t)}$ is
\[
\sum_{j\neq \hat{j}, 0\leq j\leq i-1}z_j^{(t-1)} + (i+1)x^{(t-1)} -(i+1)in_i^{(t-1)} \leq 2i < 2i+0.5 = b_{z_i}.
\]
In all three cases, we have $z_i^{(t)} = 0$.\\
\textbf{If:}
For the if direction, if $x^{(t-1)} = 1$, $in_i^{(t-1)} = 0$ and for all $j, 0\leq j<i$ we have $z_j^{(t-1)} = 1$, then the potential of $z_i^{(t)}$ is
\[
(i+1) + \sum_{j=0}^{i-1}1  + iz_i^{(t-1)}\geq 2i+1 > 2i+0.5 = b_{z_i}.
\]
If $x^{(t-1)} = 1$, $in_i^{(t-1)} = 0$ and $z_i^{(t-1)} = 1$, then the potential of $z_i^{(t)}$ is
\[
(i+1) + \sum_{j=0}^{i-1}z_j^{(t-1)} + i\geq 2i+1 > 2i+0.5 = b_{z_i}.
\]
In both cases, we have $z_i^{(t)} = 1$.

\textbf{Case (2):}
The firing rule of $in_i^{(t)}$ can be analyzed similarly. 

The potential of $in_i^{(t)}$ is
\[
\sum_{j=1}^iw_{z_jin_i}z_j^{(t-1)} = \sum_{j=1}^iz_j^{(t-1)}.
\]
\textbf{Only If:}
For the only if direction, if there exists $\hat{j}, 1\leq \hat{j}\leq i$ such that $x_{\hat{j}}^{(t-1)} = 0$, then the potential of $in_i^{(t)}$ is
\[
\sum_{j\neq \hat{j}, 1\leq j\leq i}z_j^{(t-1)}\leq i-1 < i-0.5 = b_{in_i}.
\]
We have $in_i^{(t)} = 0$.\\ 
\textbf{If:}
For the if direction, if for all $j, 1\leq j\leq i$ we have $z_j^{(t-1)} = 1$, then the potential of $in_i^{(t)}$ is
\[
\sum_{j = 1}^i1 = i > i-0.5 = b_{in_i}.
\]
We have $in_i^{(t)} = 1$ as desired.
\end{proof}

Using the above lemma, we can verify that indeed the network at the first stage fires in binary, with $z_i$ encoding the $i$th digit in the binary representation. 
\begin{theorem}
\label{binary}
Given $i\geq 1$ and $t\geq 0$, if for all $t'$ such that $0\leq t' \leq t$ we have $x^{(t')} = 1$, then
\begin{enumerate}
    \item $z_i^{(t)} = a_i$ for $t = \sum_{j=0}^\infty a_j2^j$ where $a_j\in \lbrace 0, 1\rbrace$.
    \item $in_i^{(t)} = 1$ if and only if $t \bmod 2^{i+1} = 2^{i+1} - 1$ or $0$. 
\end{enumerate}
\end{theorem}
\begin{proof}
First, let's verify that the claim is true for $z_0$. Since for all $t',0\leq t' \leq t$ we have $x^{(t')} = 1$, $z_0^{(t')} = 1$ if and only if $z_0^{(t'-1)} = 0$. This implies exactly $z_0^{(t)} = t \bmod 2$ as desired (for all the modular arithematic at this work, we choose the smallest nonnegative number from the equivalence class). Now let's do the induction on $t$ and we will verify the induction by checking $z_i, in_i$ fires in according to the induction hypothesis for all $i\geq 1$. When $t=1$, the induction statement is trivially satisfied for all $i\geq 1$. Fix $i$, we have the following cases:
\begin{enumerate}
  \item $0<t \bmod 2^{i+1} < 2^i, z_i^{(t-1)} = 0$:\\ This implies that $0\leq t-1 \bmod 2^i < 2^i-1$. By induction hypothesis, not all $z_j^{(t-1)} = 1$ for $0\leq j < i$. Now by Lemma \ref{firing1}, we have $z_i^{(t)} = 0 = a_i, in_i^{(t)} = 0$ as desired.
  \item $t \bmod 2^{i+1} = 2^i, z_i^{(t-1)} = 0, in_i^{(t-1)} = 0$:\\ This implies that $t-1 \bmod 2^i = 2^i-1$. By induction hypothesis, for all $j, 0\leq j < i$ we have $z_j^{(t-1)} = 1$. Now by Lemma \ref{firing1}, we have $z_i^{(t)} = 1 = a_i, in_i^{(t)} = 0$ as desired.
  \item $2^i<t \bmod 2^{i+1} < 2^{i+1} - 1, z_i^{(t-1)} = 1, in_i^{(t-1)} = 0$:\\ This implies that $0\leq t-1 \bmod 2^i < 2^i-2$. By induction hypothesis, not all $j, 1\leq j < i$ we have $z_j^{(t-1)} = 1$. Now by Lemma \ref{firing1}, we have $z_i^{(t)} = 1 = a_i, in_i^{(t)} = 0$ as desired.
  \item $t \bmod 2^{i+1} = 2^{i+1} - 1, z_i^{(t-1)} = 1, in_i^{(t-1)} = 0$:\\ This implies that $t-1 \bmod 2^i = 2^i-2$. By induction hypothesis, for all $j, 1\leq j < i$ we have $z_j^{(t-1)} = 1$. Now by Lemma \ref{firing1}, we have $z_i^{(t)} = 1 = a_i, in_i^{(t)} = 1$ as desired.
  \item $t \bmod 2^{i+1} = 0, z_i^{(t-1)} = 1, in_i^{(t-1)} = 1$:\\ This implies that $t-1 \bmod 2^i = 2^i-1$. By induction hypothesis, for all $j, 1\leq j < i$ we have $z_j^{(t-1)} = 1,$. Now by Lemma \ref{firing1}, we have $z_i^{(t)} = 0 = a_i, in_i^{(t)} = 1$ as desired.
\end{enumerate}
This completes the induction.
\end{proof}

\subsection{Second Stage: Capture Network} Now the second stage is a simple ``capture network" with input neurons $x$, $z_i$ for all $i, 0\leq i \leq n$, output neurons $y_i$ for $0\leq i \leq n$ and an auxilary neuron $s$. Intuitively, the network persistently captures the state of $z_i$ for all $i, 0\leq i\leq n$ into $y_i$ for all $i, 0\leq i\leq n$. We will specify the timing of the states of $z_i$ being captured later. The network is defined as the following:
\begin{multline*}
\forall\, 0\leq i\leq n, w_{xy_i} = -2, w_{y_iy_i} = 4, w_{z_iy_i} = 1, w_{z_is} = w_{y_is} = 1, w_{sy_i} = -1.5,
\end{multline*}
and
\[
w_{xs} = -n-1, w_{ss} = n+2, b_s = 0.5, \forall\, 0\leq i\leq n, b_{y_i} = 0.5.
\]
\begin{figure}
\includegraphics[width=9.5cm]{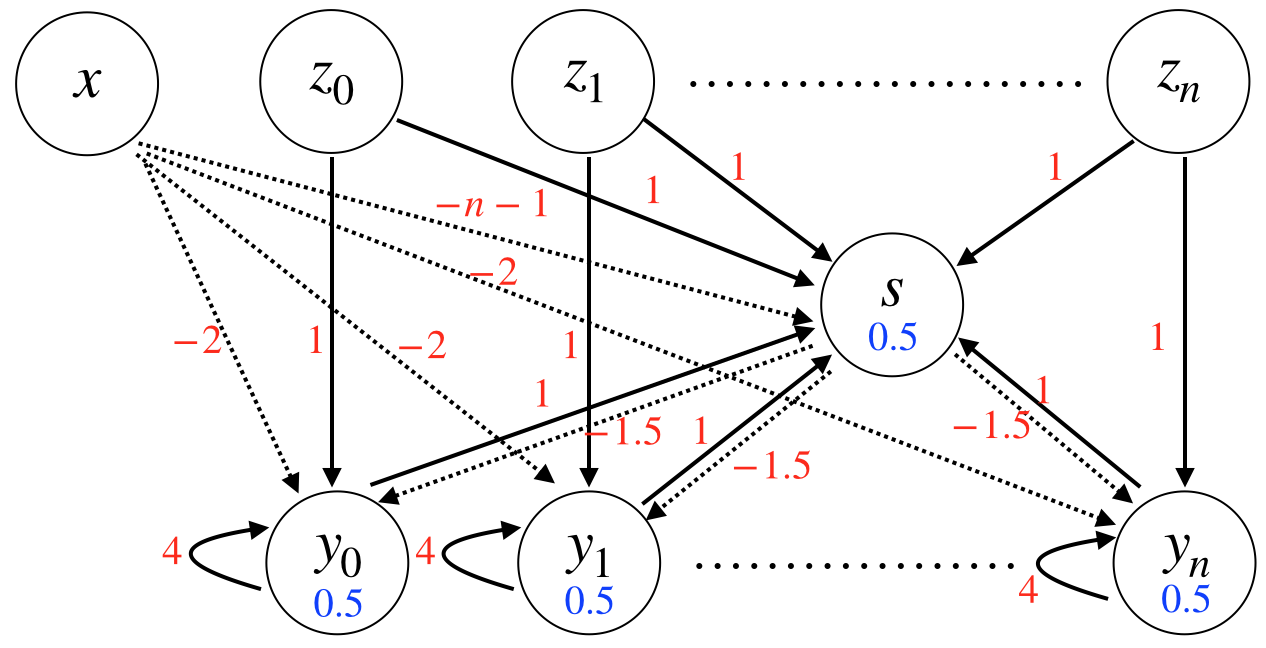}
\centering
\caption{Second Stage}
\end{figure}
Notice that the above weight ensures the following one step firing rule:
\begin{lemma}
\label{firing2}
For $0\leq i\leq n$, neurons $y_i^{(t)}, s^{(t)}$ fire according to the following rules: 
\begin{enumerate}
	\item $y_i^{(t)} = 1$ if and only if $y_i^{(t-1)} = 1$, or ($y_i^{(t-1)} = 0, x^{(t-1)} = 0, s^{(t-1)} = 0$ and $z_i^{(t-1)} = 1$).
	\item $s^{(t)} = 1$ if and only if $s^{(t-1)} = 1$, or (there exists $i, i'$ such that $z_i^{(t-1)} = 1$ or $y_{i'}^{(t-1)} = 1$, and $x^{(t-1)} = 0$).
\end{enumerate}
\end{lemma}
\begin{proof}
\textbf{Case (1):}
The potential of $y_i^{(t)}$ is
\begin{multline*}
    w_{xy_i}x^{(t-1)} + w_{y_iy_i}y_i^{(t-1)}  + w_{z_iy_i}z_i^{(t-1)} + w_{z_iz_i}z_i^{(t-1)} + w_{sy_i}s^{(t-1)} \\= -2x^{(t-1)} + 4y_i^{(t-1)} + z_i^{(t-1)} -1.5s^{(t-1)}.
\end{multline*}
\textbf{Only If:}
Let's show the only if direction for the firing rule of $y_i^{(t)}$ first. If $y_i^{(t-1)} = 0, x^{(t-1)} = 1$, the potential of $y_i^{(t)}$ is
\[
-2 + z_i^{(t-1)} -1.5s^{(t-1)} \leq -1 < 0.1 = b_{y_i}.
\] 
If $y_i^{(t-1)} = 0, s^{(t-1)} = 1$, the potential of $y_i^{(t)}$ is
\[
-2x^{(t-1)} + z_i^{(t-1)} -1.5 \leq -0.5 < 0.1 = b_{y_i}.
\] 
If $y_i^{(t-1)} = 0, z_i^{(t-1)} = 0$, the potential of $y_i^{(t)}$ is
\[
-2x^{(t-1)} -1.5s^{(t-1)} \leq 0 < 0.1 = b_{y_i}.
\] 
In all three cases, we have $y_i^{(t)} = 0$.\\
\textbf{If:}
For the if direction, if $y_i^{(t-1)} = 1$, then the potential of $y_i^{(t)}$ is
\[
-2x^{(t-1)} + 4 + z_i^{(t-1)} -1.5s^{(t-1)} \geq 0.5 > 0.1 = b_{y_i}.
\] 
If $y_i^{(t-1)} = 0, x^{(t-1)} = 0, s^{(t-1)} = 0, z_i^{(t-1)} = 1$, the potential of $y_i^{(t)}$ is
\[
4y_i^{(t-1)} + 1 \geq 1 > 0.1 = b_{y_i}.
\]
In both cases, we have $y_i^{(t)} = 1$.\\

\textbf{Case (2):}
The potential of $s^{(t)}$ is
\begin{multline*}
\sum_{j=0}^nw_{z_js}z_j^{(t-1)} + \sum_{j=0}^nw_{y_js}y_j^{(t-1)} + w_{xs}x^{(t-1)} + w_{ss}s^{(t-1)}  \\= \sum_{j=0}^nz_j^{(t-1)} + \sum_{j=0}^ny_j^{(t-1)} -(n+1)x^{(t-1)}  + (n+2)s^{(t-1)}.
\end{multline*}
\textbf{Only If:}
For the only if direction, if $s^{(t-1)} = 0$ and for all $j, 0\leq j\leq n$ we have $y_j^{(t-1)} =z_j^{(t-1)} = 0$, then the potetntial of $s^{(t)}$ is
\[
-(n+1)x^{(t-1)} \leq 0 < 0.5 = b_s.
\]
If $s^{(t-1)} = 0, x^{(t-1)} = 1$, the potetntial of $s^{(t)}$ is
\[
\sum_{j=0}^nz_j^{(t-1)} + \sum_{j=0}^nz_j^{(t-1)} -(n+1)  \leq 0 < 0.5 = b_s.
\]
In both cases, we have $s^{(t)} = 0$.\\
\textbf{If:}
For the if direction, if there exists $i, 0\leq i\leq n$ such that $y_i^{(t-1)} = 1$ and $x^{(t-1)} = 0$, then the potential of $s^{(t)}$ is
\[
\sum_{j=0}^nz_j^{(t-1)} +\sum_{j\neq i, 0\leq j\leq n}^ny_j^{(t-1)} + 1 + (n+2)s^{(t-1)}\geq 1 > 0.5 = b_s.
\]
If there exists $i, 0\leq i\leq n$ such that $z_i^{(t-1)} = 1$ and $x^{(t-1)} = 0$, the potential of $s^{(t)}$ is
\[
\sum_{j=0}^ny_j^{(t-1)} +\sum_{j\neq i, 0\leq j\leq n}^nz_j^{(t-1)} + 1 + (n+2)s^{(t-1)}\geq 1 > 0.5 = b_s.
\]
If $s^{(t-1)} = 1$, the potential of $s^{(t)}$ is
\[
\sum_{j=0}^nz_j^{(t-1)} + \sum_{j=0}^ny_j^{(t-1)} -(n+1)x^{(t-1)}  + (n+2)\geq 1 > 0.5 = b_s.
\]
In all three cases, we have $s^{(t)} = 1$ as desired.
\end{proof}
Now we can describe the behaviors of the capture network in the following theorem. The network persistantly captures the state of $z_i$ for all $i, 0\leq i \leq n$ at the first time point such that $x=0$ and there exists some $\hat{i}$ such that $z_{\hat{i}} = 1$ into $y_i$ for all $i, 0\leq i \leq n$.
\begin{theorem}
\label{capture}
For the network at the second stage, let $t'\geq 0$ be such that $x^{(t')} = 0$ and there exists $\hat{j}$ such that $z_{\hat{j}}^{(t')} = 1$, and for all $t, 0\leq t<t'$, either $x^{(t)} = 1$ or for all $i, 0\leq i\leq n$ we have $z_{i}^{(t)} = 0$. Then for all $i, t$ such that $0\leq i\leq n, t > t'$ we have $y_i^{(t)} = z_i^{(t')}$. 
\end{theorem}

\begin{proof}
First by Lemma \ref{firing2}, for all $t, 0<t\leq t'$ and for all $i, 0\leq i\leq n$ we have $y_i^{(t)} = s^{(t)} = 0$. Now at time $t'+1$, by Lemma \ref{firing2}, we see that $y_i^{(t'+1)} = z_i^{(t')}, \forall i, 0\leq i\leq n$ and $s^{(t'+1)} = 1$. Now by Lemma \ref{firing2}, we know that for all $t, t>t'$ we have $s^{(t)} = 1$. Now by Lemma \ref{firing2} again, if $y_i^{(t' + 1)} = 0$, then since for all $t, t> t'$ we have $s^{(t)} = 1$, for all $t> t'$ we have $y_i^{(t)} = 0$; and if $y_i^{(t'+1)} = 1$, then we also have for all $t, t>t'$, $y_i^{(t)} = 1$ as desired.
\end{proof}

\subsection{Wrap up}
Now we are ready to prove the main Theorem \ref{FCS} by setting $n = m = \lceil \log T'\rceil$
\begin{proof}
We are going to prove the main theorem by composing the networks from stage one and two together. If for all $t, 0\leq t\leq T$ we have $x^{(t)} = 0$, then the network satisfies the criterion trivially since for all $0\leq t\leq T$, $y_i^{(t)} = 0$. If not, then there exists $\hat{t}\geq0, L_x>0$ such that for all $t, 0\leq t<\hat{t}$ we have $x^{(t)} = 0$, for all $i, 0\leq i<L_x$ we have $x^{(\hat{t} + i)} = 1$, and $x^{(\hat{t} + L_x)} = 0$ where $L_x$ is the length of the first consecutive spikes interval. Let $L_x = \sum_{j=0}^\infty a_j2^j$; then by Theorem \ref{binary} and Lemma \ref{base}, for all $i, 0\leq i\leq n$, we have $z_i^{(\hat{t} + L_x - 1)} = a_i$. Now because $L_x>0$, we know there exists $\hat{j}$ such that $z_{\hat{j}}^{(\hat{t'} + L_x)} = 1$ by Theorem \ref{binary}. And by Lemma \ref{firing1}, we know for all $i, t$ such that $0\leq t\leq\hat{t}, 0\leq i\leq n$, we have $z_{i}^{(t)} = 0$. Now the assumption of Theorem \ref{capture} is satisfied with $t' = \hat{t} + L_x$. By Theorem \ref{capture}, we get for all $t, i$ such that $0\leq i\leq n, t\geq \hat{t} + L_x$ we have $y_i^{(t)} = a_i$ and $T+1\geq \hat{t} + L_x$ as desired. This shows that the above network solves FCSC(T) problem in time $1$ with $O(\log T)$ neurons.
\end{proof}

Notice that in fact by the proof above, FCSC network enjoys an early convergence property. The network actually converges at time $\hat{t} + L_x$. Therefore we have the following stronger version of Theorem \ref{FCS}.
\begin{corollary}
For all $t, 0\leq t\leq T$, FCSC network with $O(\log T)$ neurons solves FCSC(t) problem in time $1$.  
\end{corollary}

\section{Total Spikes Counting}
To count the total number of spikes in an arbitrary interval requires the persistence of neurons without external spikes. Notice that on FCSC network, each neuron toggles itself according to binary representation without delay. However, the persistence of neurons and toggles without delays are conflicting objectives; persistence of neurons stabilizes the network while toggling without delays changes the firing patterns of the network. For example, we use self-inhibition to count mod $2$ but if we use self-inhibition to count mod $2$, the neuron cannot maintain the count during intervals with no inputs. In this section, we circumvent this difficulty by allowing the network to enter an unstable intermediate state that still stores the information of the count when the spikes arrive; however, the network will converge to a \textit{clean state} that according to binary representation after one step of computation without external signals, and this \textit{clean state} is stable in an arbitrary interval with no input.  

In this section, because the self-inhibition used in Section $3$ to count mod $2$ cannot induce persistence, we build a network of four neurons to count mod $4$ to replace the function of $z_0, z_1$ in Section $3$. We then iteratively build the rest of the network that approximately fires in binary on top of the mod $4$ counter network.

\subsection{Mod $4$ Counter Network}
The construction of the mod $4$ counter network is the following:
\begin{multline*}
    w_{xf_i} = 1, w_{f_if_i} = 2, 0\leq i\leq 3, w_{f_{j+1}f_{j}} = -3, 0\leq j \leq 2,\\  w_{f_1f_{2}} = w_{f_2f_{3}} = w_{f_3f_0} = 1, w_{f_0f_3} = w_{f_3f_1}  = -3
\end{multline*}
and 
\[
b_{f_1} = 0.5, b_{f_i} = 1.5, i\neq 1.
\]
\begin{figure}
\includegraphics[width=9.5cm]{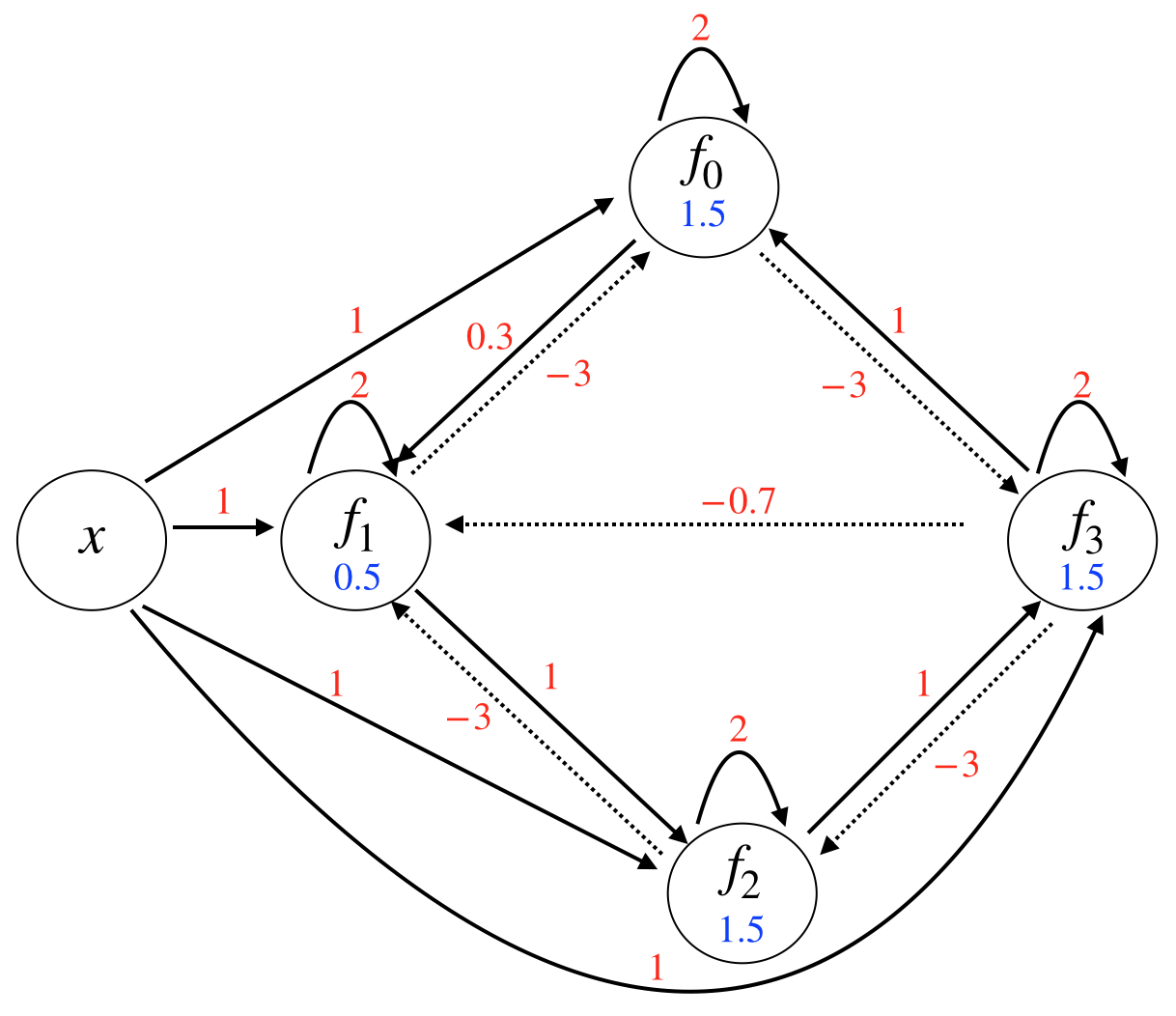}
\centering
\caption{mod $4$ Counter Network}
\end{figure}
We have the following lemma to specify the firing rules of $f_i$:
\begin{lemma}
\label{firing3}
For all $t, i$ such that $t\geq 1, 0\leq i < 4$, neurons $f_i^{(t)}$ fire according to the following rules: 
\begin{enumerate}
	\item $f_1^{(t)} = 1$ if and only if $f_2^{(t-1)}=0$, and ($x^{(t-1)} = 1, f_3^{(t-1)} = 0$ or $f_1^{(t-1)} = 1$ or $x^{(t-1)} = 1, f_0^{(t-1)} = 1$).
	\item For $i\neq 1$ we have $f_i^{(t)} = 1$ if and only if $f_{(i+1)\bmod 4}^{(t-1)} = 0$, and ($x^{(t-1)} = 1, f_{(i-1)\bmod 4}^{(t-1)} = 1$ or $f_i^{(t-1)} = 1$).
\end{enumerate}
\end{lemma}

\begin{proof}
\textbf{Case (1):}
The potential of $f_1^{(t)}$ is
\begin{multline*}
    w_{xf_1}x^{(t-1)} + w_{f_1f_1}f_1^{(t-1)} + w_{f_2f_1}f_2^{(t-1)} + w_{f_3f_1}f_3^{(t-1)} + w_{f_0f_1}f_0^{(t-1)} \\= x^{(t-1)} + 2f_1^{(t-1)} -3f_2^{(t-1)} -0.7f_3^{(t-1)} + 0.3f_0^{(t-1)}.
\end{multline*}
\textbf{Only If:}
Let's show the only if direction for the firing rule of $f_1^{(t)}$ first. If $f_2^{(t-1)} = 1$, then the potential of $f_1^{(t)}$ is
\[
x^{(t-1)} + 2f_1^{(t-1)} -3 -0.7f_3^{(t-1)} + 0.3f_0^{(t-1)}\leq 0.3 <0.5 = b_{f_1}.
\] 
If $f_1^{(t-1)} = 0, x^{(t-1)} = 0$, then the potential of $f_1^{(t)}$ is
\[
 -3f_2^{(t-1)} -0.7f_3^{(t-1)} + 0.3f_0^{(t-1)}\leq 0.3 <0.5 = b_{f_1}.
\] 
If $f_1^{(t-1)} = 0, f_3^{(t-1)} = 1, f_0^{(t-1)} = 0$, then the potential of $f_1^{(t)}$ is
\[
x^{(t-1)}  -3f_2^{(t-1)} -0.7 \leq 0.3 <0.5 = b_{f_1}.
\] 
In all three cases, we have $f_1^{(t)} = 0$.\\
\textbf{If:}
For the if direction, if $f_2^{(t-1)} =0, f_1^{(t-1)} = 1$, then the potential of $f_1^{(t)}$ is
\[
x^{(t-1)} + 2  -0.7f_3^{(t-1)} + 0.3f_0^{(t-1)} \geq 1.3 >0.5 = b_{f_1}.
\] 
If $f_2^{(t-1)} =0, x^{(t-1)} = 1, f_3^{(t-1)} = 0$, then the potential of $f_1^{(t)}$ is 
\[
1+ 2f_1^{(t-1)}  + 0.3f_0^{(t-1)} \geq 1 >0.5 = b_{f_1}.
\] 
If $f_2^{(t-1)} =0, x^{(t-1)} = 1, f_0^{(t-1)} = 1$, then the potential of $f_1^{(t)}$ is 
\[
1 + 2f_1^{(t-1)}  -0.7f_3^{(t-1)} + 0.3\geq 0.6 >0.5 = b_{f_1}.
\] 
In all three cases, we have $f_1^{(t)} = 1$. 

\textbf{Case (2):}
For $i\neq 1$, The potential of $f_i^{(t)}$ is
\begin{multline*}
w_{xf_i}x^{(t-1)} + w_{f_if_i}f_i^{(t-1)} + w_{f_{(i-1)\bmod 4}f_{i}}f_{(i-1)\bmod 4}^{(t-1)} + w_{f_{(i+1)\bmod 4}f_i}f_{(i+1)\bmod 4}^{(t-1)} \\= x^{(t-1)} + 2f_i^{(t-1)} +f_{(i-1)\bmod 4}^{(t-1)} -3f_{(i+1)\bmod 4}^{(t-1)}.
\end{multline*}
\textbf{Only If:}
For the only if direction, if $f_{(i+1)\bmod 4}^{(t-1)} = 1$, then the potential of $f_i^{(t)}$ is
\[
x^{(t-1)} + 2f_i^{(t-1)} +f_{(i-1)\bmod 4}^{(t-1)} -3 \leq 1 < 1.5 = b_i.
\]
If $x^{(t-1)} = 0, f_i^{(t-1)} = 0$, then the potential of $f_i^{(t)}$ is
\[
f_{(i-1)\bmod 4}^{(t-1)} -3f_{(i+1)\bmod 4}^{(t-1)} \leq 1 < 1.5 = b_i.
\]
If $f_{(i-1)\bmod 4}^{(t-1)}  = 0, f_i^{(t-1)} = 0$, then the potential of $f_i^{(t)}$ is
\[
x^{(t-1)} -3f_{(i+1)\bmod 4}^{(t-1)} \leq 1 < 1.5 = b_i.
\]
In all three cases, we have $f_i^{(t)} = 0$.\\
\textbf{If:}
For the if direction, if $f_{(i+1)\bmod 4}^{(t-1)} = 0, x^{(t-1)} = 1, f_{(i-1)\bmod 4}^{(t-1)} = 1$, then the potential of $f_i^{(t)}$ is
\[
1 + 2f_i^{(t-1)} + 1 \geq 2 > 1.5 = b_i.
\]
If $f_{(i+1)\bmod 4}^{(t-1)} = 0, f_i^{(t-1)} = 1$, then the potential of $f_i^{(t)}$ is
\[
x^{(t-1)} + 2 +f_{(i-1)\bmod 4}^{(t-1)}  \geq 2 > 1.5 = b_i.
\]
In both cases, we have $f_i^{(t)} = 1$ as desired.
\end{proof}

For $0\leq i<4$, define a \textit{clean state} with value $i$ at time $t'$ of the mod $4$ counter network to be a state in which $f_i^{(t')} = 1$ and for all $j, j\neq i$ we have $f_j^{(t')} = 0$. By Lemma \ref{firing3}, it is trivial to see that if for all $t,t\geq t'$ we have $x^{(t)} = 0$, then for all $t,t\geq t'$ and for all $i, 0\leq i < 4$ we have $f_i^{(t)} = f_i^{(t')}$. Using Lemma \ref{firing3}, we have the following lemma describing the behaviors of mod $4$ counter network. Intuitively, when a new input arrives, the network enters an intermediate state in which both neurons represent the old count and the new count fire; when there is no input, the neuron that represents the new count will inhibit the neuron that represents the old count to stabilize the network in a \textit{clean state}. 
\begin{lemma}
\label{mod4}
Let the mod $4$ counter network be at a clean state with value $\hat{i}$ at time $t'$. Fix a positive integer $L$. For all $i, 0\leq i< L$, let $x^{(t' + i)} = 1$ and $x^{(t' + L)} = 0$. Then, at time $t, t'< t < t'+L+1$, we have the state of the network being
\[
f_{(\hat{i} + t - t')\bmod 4}^{(t)} = f_{(\hat{i} + t - t' - 1)\bmod 4}^{(t)} = 1, f_{(\hat{i} + t - t' - 2)\bmod 4}^{(t)} = f_{(\hat{i} + t - t' - 3)\bmod 4}^{(t)} = 0.
\]
Furthermore, the network will be at a clean state again at time $t'+L+1$ with $f_{(\hat{i} + L)\bmod 4}^{(t'+L+1)} = 1$.
\end{lemma}

\begin{proof}
First, let's use induction on $t$ to prove at time $t, t'< t < t'+L+1$, we have the state of the network be 
\[
f_{(\hat{i} + t - t')\bmod 4}^{(t)} = f_{(\hat{i} + t - t' - 1)\bmod 4}^{(t)} = 1, f_{(\hat{i} + t - t' - 2)\bmod 4}^{(t)} = f_{(\hat{i} + t - t' - 3)\bmod 4}^{(t)} = 0.
\]
\textbf{Base Case:}
By Lemma \ref{firing3}, we have 
\[
f_{(\hat{i} + 1)\bmod 4}^{(t'+1)} = f_{(\hat{i} + t - t')\bmod 4}^{(t'+1)} = 1, f_{(\hat{i} -1)\bmod 4}^{(t'+1)} = f_{(\hat{i} - 2)\bmod 4}^{(t'+1)} = 0
\]
for the base case.\\
\textbf{Inductive Step:}
Now assume the induction hypothesis is true for $t = k$, since we have $x^{(k)} = 1$ by Lemma \ref{firing3}, we indeed have
\[
f_{(\hat{i} + k+1 - t')\bmod 4}^{(k+1)} = f_{(\hat{i} + k+1 - t' - 1)\bmod 4}^{(k+1)} = 1, f_{(\hat{i} + k+1 - t' - 2)\bmod 4}^{(k+1)} = f_{(\hat{i} + k+1 - t' - 3)\bmod 4}^{(k+1)} = 0.
\]
This completes the induction.

Now since $x^{(t'+L)} = 0$, by Lemma \ref{firing3} we can derive the state of the network at time $t'+L+1$
\[
f_{(\hat{i} + L)\bmod 4}^{(t'+L+1)} = 1, f_{j}^{(t'+L+1)} = 0,\ \forall j\neq (\hat{i} + L)\bmod 4
\]
as desired.
\end{proof}

\subsection{TSC Network}
Now we iteratively build the network with the following rule on top of the mod $4$ counter network,
\begin{multline*}
w_{f_3z_i} = w_{f_3in_i} =  3, w_{f_0z_i} = w_{f_0in_i} = -1, w_{xz_i} = w_{xin_i} = 1,\\ w_{z_jz_i} =  w_{z_jin_i} = 1,\ \forall j, 2\leq j<i, w_{in_iz_i} = -i-3, w_{z_iin_i} = 1, w_{z_iz_i} = i + 3
\end{multline*}
and
\[
b_{z_i} = i+1.5,\ b_{in_i} = i + 2.5.
\]
\begin{figure}
\centering
\includegraphics[width=10.3cm]{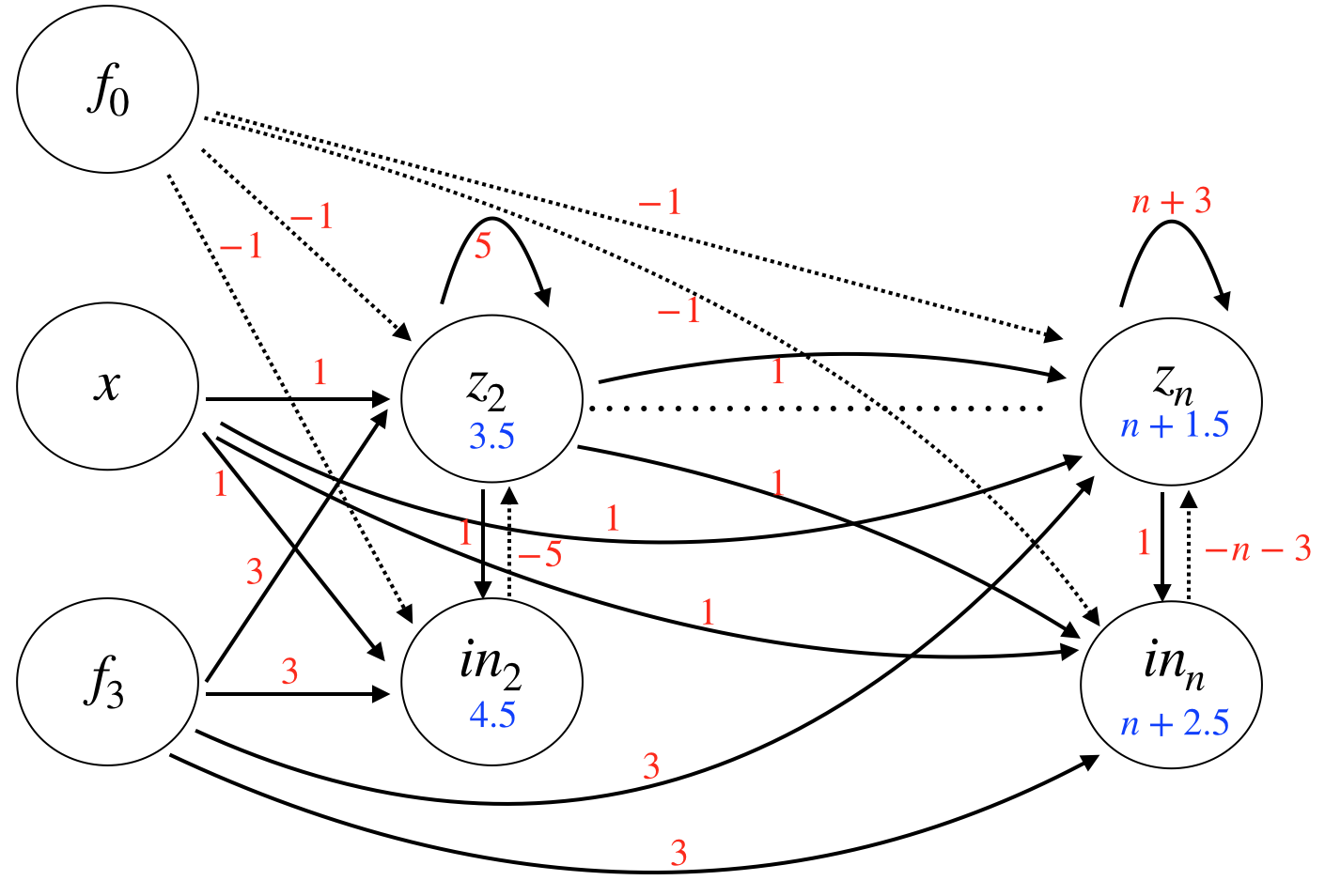}
\caption{Total spikes counting (TSC) Network}

\end{figure}
In the full construction of the TSC network, intuitively, we replace the function of $z_0, z_1$ in Section $3$ with a mod $4$ counter network. We design the weights coming from $f_3, f_0$ such that they will induce proper carry in an approximate binary representation at $z_i, i\geq 2$, and we use a similar idea as the mod $4$ counter network to make TSC network converge to an exact binary representation in one computation step without input. 

The following lemma specifies the firing rules of $z_i, in_i$ for $i\geq 2$:
\begin{lemma}
\label{firing4}
For $i\geq 2$, neurons $z_i^{(t)}, in_i^{(t)}$ fire according to the following rules: 
\begin{enumerate}
	\item $z_i^{(t)} = 1$ if and only if $in_i^{(t-1)} = 0$, and either ($f_3^{(t-1)} = 1, f_0^{(t-1)} = 0, x^{(t-1)} = 1$ and for all $j, 2\leq j<i$ we have $z_j^{(t-1)} = 1$) or $z_i^{(t-1)} = 1$.
	\item $in_i^{(t)} = 1$ if and only if $z_i^{(t-1)} = 1, f_3^{(t-1)} = 1, f_0^{(t-1)} = 0, x^{(t-1)} = 1$ and for all $j, 2\leq j<i$ we have $z_j^{(t-1)} = 1$.
\end{enumerate}
\end{lemma}

\begin{proof}
\textbf{Case (1):}
The potential of $z_i^{(t)}$ is

\begin{multline*}
w_{f_3z_i}f_3^{(t-1)} + w_{f_0z_i}f_0^{(t-1)} + \sum_{j=2}^{i-1}w_{z_jz_i}z_j^{(t-1)} + w_{z_iz_i}z_i^{(t-1)} + w_{in_iz_i}in_i^{(t-1)} + w_{xz_i}x^{(t-1)} \\= 3f_3^{(t-1)} -f_0^{(t-1)} + \sum_{j=2}^{i-1}z_j^{(t-1)} + (i+3)z_i^{(t-1)} -(i+3)in_i^{(t-1)} + x^{(t-1)}.
\end{multline*}
\textbf{Only If:}
Let's show the only if direction for the firing rule of $z_i^{(t)}$ first. If $in_i^{(t-1)} = 1$, the potential of $z_i^{(t)}$ is
\[
3f_3^{(t-1)} -f_0^{(t-1)} + \sum_{j=2}^{i-1}z_j^{(t-1)} + (i+3)z_i^{(t-1)} -(i+3)+ x^{(t-1)} \leq i+1 < i+1.5 = b_{z_i}.
\] 
If $f_3^{(t-1)} = 0, z_i^{(t-1)} = 0$, the potential of $z_i^{(t)}$ is
\[
-f_0^{(t-1)} + \sum_{j=2}^{i-1}z_j^{(t-1)}  -(i+3)in_i^{(t-1)}+ x^{(t-1)} \leq i-1 < i+1.5 = b_{z_i}.
\] 
If $f_0^{(t-1)} = 1, z_i^{(t-1)} = 0$, the potential of $z_i^{(t)}$ is
\[
3f_3^{(t-1)} -1 + \sum_{j=2}^{i-1}z_j^{(t-1)}  -(i+3)in_i^{(t-1)}+ x^{(t-1)} \leq i+1 < i+1.5 = b_{z_i}.
\] 
If $x^{(t-1)} = 0, z_i^{(t-1)} = 0$, the potential of $z_i^{(t)}$ is
\[
3f_3^{(t-1)} -f_0^{(t-1)} + \sum_{j=2}^{i-1}z_j^{(t-1)}  -(i+3)in_i^{(t-1)}\leq i+1 < i+1.5 = b_{z_i}.
\] 
If $z_i^{(t-1)} = 0$ and there exists $\hat{j}, 2\leq \hat{j}<i$ such that $z_{\hat{j}}^{(t-1)} = 0$, the potential of $z_i^{(t)}$ is
\[
3f_3^{(t-1)} -f_0^{(t-1)} + \sum_{j\neq \hat{j}, 2\leq j<i}z_j^{(t-1)}  -(i+3)in_i^{(t-1)}+ x^{(t-1)} \leq i+1 < i+1.5 = b_{z_i}.
\] 
In all cases, we have $z_i^{(t)} = 0$.\\
\textbf{If:}
For the if direction, if $in_i^{(t-1)} = 0, f_3^{(t-1)} = 1, f_0^{(t-1)} = 0, x^{(t-1)} = 1$ and for all $j, 2\leq j<i$ we have $z_j^{(t-1)} = 1$, then the potential of $z_i^{(t)}$ is
\[
3  + \sum_{j=2}^{i-1}1+ (i+3)z_i^{(t-1)} + 1 \geq i+2 > i+1.5 = b_{z_i}.
\] 
If $in_i^{(t-1)} = 0, z_i^{(t-1)} = 1$, the potential of $z_i^{(t)}$ is
\[
3f_3^{(t-1)} -f_0^{(t-1)} + \sum_{j=2}^{i-1}z_j^{(t-1)} + (i+3) + x^{(t-1)}\geq i+2 > i+1.5 = b_{z_i}.
\]
In both cases, we have $z_i^{(t)} = 1$.

\textbf{Case (2):} The potential of $in_i^{(t)}$ is
\begin{multline*}
w_{f_3in_i}f_3^{(t-1)} + w_{f_0in_i}f_0^{(t-1)} + \sum_{j=2}^{i-1}w_{z_jin_i}z_j^{(t-1)} + w_{z_iin_i}z_i^{(t-1)}  + w_{xin_i}x^{(t-1)} \\= 3f_3^{(t-1)} -f_0^{(t-1)} + \sum_{j=2}^{i-1}z_j^{(t-1)} + z_i^{(t-1)} + x^{(t-1)}.
\end{multline*}
\textbf{Only If:}
For the only if direction, if $z_i^{(t-1)} = 0$, then the potential of $in_i^{(t)}$ is 
\[
3f_3^{(t-1)} -f_0^{(t-1)} + \sum_{j=2}^{i-1}z_j^{(t-1)} + x^{(t-1)} \leq i+2 < i+2.5 = b_{in_i}.
\]
If $f_3^{(t-1)} = 0$, the potential of $in_i^{(t)}$ is 
\[
-f_0^{(t-1)} + \sum_{j=2}^{i-1}z_j^{(t-1)} + z_i^{(t-1)} + x^{(t-1)} \leq i < i+2.5 = b_{in_i}.
\]
If $f_0^{(t-1)} = 1$, the potential of $in_i^{(t)}$ is 
\[
3f_3^{(t-1)} -1 + \sum_{j=2}^{i-1}z_j^{(t-1)} + z_i^{(t-1)} + x^{(t-1)} \leq i+2 < i+2.5 = b_{in_i}.
\]
If $x^{(t-1)} = 0$, the potential of $in_i^{(t)}$ is 
\[
3f_3^{(t-1)} -f_0^{(t-1)} + \sum_{j=2}^{i-1}z_j^{(t-1)} + z_i^{(t-1)}  \leq i+2 < i+2.5 = b_{in_i}.
\]
If there exists $\hat{j}, 2\leq \hat{j} < i$ such that $z_{\hat{j}} = 0$, the potential of $in_i^{(t)}$ is 
\[
3f_3^{(t-1)} -f_0^{(t-1)} + \sum_{j\neq \hat{j}, 2\leq j < i}z_j^{(t-1)} + z_i^{(t-1)} + x^{(t-1)}  \leq i+2 < i+2.5 = b_{in_i}.
\]
In all cases, $in_i^{(t)} = 0$.\\
\textbf{If:}
For the if direction, if $z_i^{(t-1)} = 1, f_3^{(t-1)} = 1, f_0^{(t-1)} = 0, x^{(t-1)} = 1$ and for all $j, 2\leq j<i$ we have $z_j^{(t-1)} = 1$, then the potential of $in_i^{(t)}$ is 
\[
3+ \sum_{j=2}^{i-1}1 + 1+ 1  \leq i+3 > i+2.5 = b_{in_i}.
\]
We have $in_i^{(t)} = 1$ as desired.
\end{proof}

Define a clean state at time $t'$ of TSC network with value $X$ stored be one in which
\begin{enumerate}
	\item $f_{X\bmod 4}^{(t')} = 1, f_j^{(t')} = 0,\ \forall j\neq X\bmod 4$ (i.e., the mod $4$ counter subnetwork is clean with value $X\bmod 4$).
	\item For $X = \sum_{i=0}^\infty a_i2^i, a_i\in \lbrace 0, 1\rbrace$, $z_{k}^{(t')} = a_k,\forall k\geq 2$.
	\item $in_i^{(t')} = 0\text{ if }X\bmod 2^{i+1} = 2^{i+1}-1$.
\end{enumerate}
So being at a clean state for TSC network with value $X$ stored implies being at a clean state with value $X\bmod 4$ for its mod $4$ counter subnetwork with $z_i$ in binary representation for $i\geq 2$. By Lemma \ref{firing4}, it is trivial to see that if for all $t\geq t'$ we have $x^{(t)} = 0$, then for all $i\geq 2$ and for all $t, t\geq t'$ we have $f_i^{(t)} = f_i^{(t')}$. Using Lemma \ref{firing4}, we have the following lemma describing the behaviors of the TSC network.
\begin{lemma}
\label{TSC}
Let TSC network be at a clean state at time $t'$ with value $X$ stored. Fix a positive integer $L$. For all $i$ such that $0\leq i< L$, let $x^{(t' + i)} = 1$ and $x^{(t' + L)} = 0$. Then, at $t, t'<t<t'+L+1$, $z_i, in_i$ fire with the following rules for all $i\geq 2$:
\begin{enumerate}
	\item for $1=X + t -t' \bmod 2^{i+1} < 2^i$, $z_i^{(t)} = 0$.
	\item for $1 <X + t -t' \bmod 2^{i+1} < 2^i$, $z_i^{(t)} = in_i^{(t)} = 0$.
	\item for $X + t-t' \bmod 2^{i+1} \geq 2^i$, we have $z_i^{(t)} = 1, in_i^{(t)} = 0$.
	\item for $X + t-t' \bmod 2^{i+1} = 0$, we have $z_i^{(t)} = 1, in_i^{(t)} = 1$.
\end{enumerate}
Furthermore, the network will be at a clean state with value $X+L$ stored at time $t'+L+1$. 
\end{lemma}

\begin{proof}
Just like the mod $4$ counter network case, we want to deduce the behaviors of network at $t, t'<t<t'+L+1$ using induction first.\\ 
\textbf{Base Case:}
Fix $i$, for $t = t'+1$, we have the following cases
\begin{enumerate}
  \item $0 <X + 1 \bmod 2^{i+1} < 2^i$:\\ This implies that $0\leq X \bmod 2^{i+1}< 2^i-1$. This shows that not all $j, j< i$ we have $z_j^{(t-1)} = 1$ or $f_3^{(t-1)} = 0$ or $f_0^{(t-1)} = 1$. By Lemma \ref{firing4}, we have $z_i^{(t)} = in_i^{(t)} = 0$.
  \item $X + 1 \bmod 2^{i+1} \geq 2^i$:\\ This implies that $2^i-1\leq X \bmod 2^{i+1}< 2^{i+1}-1$. This shows that either for all $j, j< i$ we have $f_3^{(t-1)} = 1, f_0^{(t-1)} = 0, z_j^{(t-1)} = 1$ or $z_i^{(t-1)} = 1$ but not both. By Lemma \ref{firing4}, we have $z_i^{(t)} = 1, in_i^{(t)} = 0$.
  \item $X + 1 \bmod 2^{i+1} = 0$:\\ This implies that $X\bmod 2^{i+1} = 2^{i+1}-1$. This shows that $f_3^{(t-1)} = 1, f_0^{(t-1)} = 0$ and for all $j\leq i$ we have $z_j^{(t-1)} = 1$ and by the definition of a clean state, we have $in_i^{(t-1)} = 0$. Now by Lemma \ref{firing4}, we have $z_i^{(t)} = 1, in_i^{(t)} = 1$.
\end{enumerate}
\textbf{Inductive Step:}
Assume the induction hypothesis is accurate for $t = k$. We have the following cases
\begin{enumerate}
  \item $1=X + k+1 - t'\bmod 2^{i+1} < 2^i$:\\ This implies that $X + k - t'\bmod 2^{i+1} = 0$. Now by induction hypothesis and Lemma \ref{mod4}, we know that $f_3^{(k)} =1, f_0^{(k)} = 0$ and for all $j, i\geq j\geq 2$ we have $z_j^{(k)} = 1, in_j^{(k)} = 1$. By Lemma \ref{firing4}, we have $z_i^{(k+1)} = 0, in_i^{(k+1)} = 1$.
  \item $1 <X + k+1 - t' \bmod 2^{i+1} < 2^i$:\\ This implies that $1\leq X + k - t' \bmod 2^{i+1}< 2^i-1$. By induction hypothesis and Lemma \ref{mod4}, this shows that not all $j, j< i$ we have $z_j^{(k)} = 1$ or $f_3^{(k)} = 0$ or $f_0^{(k)} = 1$. By Lemma \ref{firing4}, we have $x_i^{(k+1)} = in_i^{(k+1)} = 0$.
  \item $X + k+1 - t' \bmod 2^{i+1} \geq 2^i$:\\ This implies that $2^i-1\leq X + k - t' \bmod 2^{i+1}< 2^{i+1}-1$. By induction hypothesis and Lemma \ref{mod4}, this shows that either for all $j, j< i$ we have $f_3^{(k)} = 1, f_0^{(k)} = 0,z_j^{(k)} = 1$ or $z_i^{(k)} = 1$ but not both. By Lemma \ref{firing4}, we have $z_i^{(k+1)} = 1, in_i^{(k+1)} = 0$.
  \item $X + k+1 - t' \bmod 2^{i+1} = 0$:\\ This implies that $X + k - t'\bmod 2^{i+1} = 2^{i+1}-1$. By induction hypothesis and Lemma \ref{mod4}, this shows that all $f_3^{(k)} = 1, f_0^{(k)} = 0, in_i^{(k)} = 0$ and for all $j, j\leq i$ we have $ z_j^{(k)} = 1$. Now by Lemma \ref{firing4}, we have $z_i^{(t)} = 1, in_i^{(t)} = 1$.
\end{enumerate}
This completes the induction.

Now we just need to show that at time $t'+L+1$ the network is at a clean state with value $X+L$ stored. We have the following cases:
\begin{enumerate}
  \item $1=X + L \bmod 2^{i+1} < 2^i$:\\
  By above induction, we have for $j, j\leq i$, $z_j^{(t'+L)} = 0$. No matter what the value of $in_i^{(t'+L)}$ is, by Lemma \ref{firing4} we have $z_i^{(t'+L+1)} = in_i^{(t'+L+1)} = 0$.
  \item $1 <X + L \bmod 2^{i+1} < 2^i$, $z_i^{(t)} = in_i^{(t)} = 0$:\\
  By above induction, we have $z_i^{(t'+L)} = in_i^{(t'+L)} = 0$. By Lemma \ref{firing4}, we have $z_i^{(t'+L+1)} = in_i^{(t'+L+1)} = 0$.
  \item $X + L \bmod 2^{i+1} \geq 2^i$, we have $z_i^{(t'+L)} = 1, in_i^{(t'+L)} = 0$. By Lemma \ref{firing4}, we have $z_i^{(t'+L+1)} = in_i^{(t'+L+1)} = 0$.
  \item $X + L \bmod 2^{i+1} = 0$, we have $z_i^{(t'+L)} = 1, in_i^{(t'+L)} = 1$. By Lemma \ref{firing4}, we have $z_i^{(t'+L+1)} = 0, in_i^{(t'+L+1)} = 1$.
\end{enumerate}
which is exactly a clean state with value $X+L$ stored combining with Lemma \ref{mod4}.
\end{proof}

\subsection{Wrap up}
Now we are ready for the main proof of Theorem \ref{TSCMain} by setting $n = \lceil \log T'\rceil$ and let $f_i, z_j, 0\leq i\leq 3, 2\leq j\leq n$ be our output neurons.
\begin{proof}
Let $f_i, z_j, 0\leq i < 4, 2\leq j\leq n$ be our output neurons. Let there be $X$ spikes in $T$ time steps. Let $[t_0, t_0+X_0-1], \dotsb, [t_k, t_k + X_k - 1]$ be the disjoint maximal intervals of spikes ordered by time (i.e., $x^{(t)} = 1$ if $t\in [t_i, t_i + X_i - 1]$ for some $0\leq i\leq k$ and $[t_i, t_i + X_i]\cap [t_j, t_j + X_j] = \emptyset$ for all $i\neq j$ and $t_0<t_1<\dotsb < t_k$, $\sum_{i=0}^kX_k = X$). Now I claim that at time $t_i + X_i+1$, the network is at a clean state with value $\sum_{j=0}^iX_j$ stored. We will prove the claim with induction on $i$. For $i = 0$, apply Lemma \ref{TSC}, we get that the network is at a clean state with value $X_0$ stored. Assume the network is at a clean state with value $\sum_{j=0}^iX_j$ stored at time $t_i + X_i + 1$. Then apply Lemma \ref{TSC} again, we get at time $t_{i+1} + X_{i+1} + 1$, the network is at a clean state with value $\sum_{j=0}^{i+1}X_j$ stored at time $t_{i+1} + X_{i+1} + 1$. So at time $t_k + X_k + 1 \leq T+1$, the network is at a clean state with value $\sum_{j=0}^{k}X_j = X$ stored as desired. This shows that the above network solves TSC(T) problem in time $1$ with $O(\log T)$ neurons.
\end{proof}

Notice that in fact by the proof above, TSC network enjoys an early convergence property. The network actually converges at time $t_k + X_k + 1$. Therefore we have the following stronger version of Theorem \ref{TSCMain}.
\begin{corollary}
For all $t, 0\leq t\leq T$, TSC network with $O(\log T)$ neurons solves FCSC(t) problem in time $1$.  
\end{corollary}

\section{Time Lower Bound for FCSC and TSC}
In Section $4$, we mentioned that there is a conflicting objective between stabilizing the output and toggling without delays. We therefore introduced the idea of carrying information of the count at an unclean state and then converging to a clean state, which introduces one time step of delay. In this Section, we are going to show that this delay is unavoidable.  

Intuitively, the proof of the time lower bound uses the fact that if the network has to solve the problem without delay, the network must stabilize immediately at each time step. Therefore, the neurons that fire at the last round will stay firing. By injectivity of the representation, we can conclude that the network can at most count up to the network size.

The proof of Theorem \ref{FCSClower} is the follows. The proof of Theorem \ref{TSClower} is identical.
\begin{proof}
Consider the following input sequence such that for all $0\leq t < T$ we have $x^{(t)} = 1$ and for all $t \geq T$ we have $x^{(t)} = 0$. Let $X$ be the collections of all neurons in the network. Assume for all $0\leq t\leq T$, the network solves FCSC(t) at time $0$. For all $0\leq j\leq T$, let $S_j = \lbrace y_i: y_i^{(j)} = 1, 1\leq i\leq m\rbrace$. We want to show that $S_T\supsetneq S_{T-1}\supsetneq \dotsb\supsetneq S_0$. To prove this by induction on $t$, we strengthen our induction hypothesis to become $S_t\supsetneq S_{t-1}\supsetneq \dotsb\supsetneq S_0$ and for all $y_j\in S_{t-1}$ we have $w_{xy_j} > 0$.\\ 
\textbf{Base Case:} When $t = 1$, notice that $S_0 = \emptyset$ by construction. Now by injectivity of the counter representation, we have $S_1\supsetneq S_0$ and for $y_j\in S_{0}$, $w_{xy_j} > 0$ is vacuously true.\\
\textbf{Induction Step:} Now assume $S_t\supsetneq S_{t-1}\supsetneq \dotsb\supsetneq S_1$ and $w_{xy_j} > 0$ for $y_j\in S_{t-1}$. At time step $t+1$, since the network solves FCSC(t) at time $0$, the neurons in $y$ is stabilized even without the input from $x$. This means that
\[
\sum_{z\in X/\lbrace x\rbrace}w_{zy_j}z^{(t)} - b_{y_j} > 0 \text{ if }y_j\in S_t
\]
Now since $w_{xy_j} > 0$, we know that neurons in $S_{t-1}$ will keep firing at time $t+1$. For neurons in $S_t / S_{t-1}$, since those neurons fire at time $t$, we have
\[
w_{xy_j} + \sum_{z\in X/\lbrace x\rbrace}w_{zy_j}z^{(t-1)} - b_{y_j} > 0 \text{ if }y_j\in S_t/S_{t-1}
\]
And since the network solves FCSC(t-1) at time $t-1$, we also have
\[
\sum_{z\in X/\lbrace x\rbrace}w_{zy_j}z^{(t-1)} - b_{y_j} \leq 0 \text{ if }y_j\in S_t/S_{t-1}
\]
Substract two equations we get
\[
w_{xy_j} > 0 \text{ if }y_j\in S_t/S_{t-1}
\]
And hence $S_{t+1}\subset S_t$. By injectivity of the count representation, we have $S_{t+1}\supsetneq S_t$ as desired. 

Now we have $S_{T}\supsetneq S_{T'-1}\supsetneq\dotsb\supsetneq S_2\supsetneq S_1$, but we only have less than $T$ neurons. Contradiction.
\end{proof}

\section{Discussion and Future Directions}
In this work, we model how brains process temporal information over a long time range using neurons with transient activities. We propose two tasks that correspond to two common neural coding schemes, temporal coding and rate coding. ``First consecutive spikes counting" (FCSC) is equivalent to counting the distance between the first two spikes, a prevalent temporal coding scheme in the sensory cortex while “Total spikes counting"(TSC) counts the number of the spikes over an arbitrary interval, which is an example of a rate coding. We design two networks with memoryless neurons that solve the above two problems in time $1$ with $O(\log T)$ neurons and show that the time bound is tight. 

A natural extension is to consider general temporal coding. Instead of coding the distance between the first two spikes, we code an arbitrary spike pattern within a max input interval length $T$. This can be done as an application of the FCSC network. Given an input pattern with $K$ spikes, we can count the spike interval between each pair of spikes using an FCSC network and have a network that processes an arbitrary spike pattern with $K$ spikes in time $1$ with $O(K\log T)$ neurons. Since typically the temporal coding in the brain is sparse, we have $K\ll T$ and therefore the network 
only uses a small number of neurons.

Out of the spiking neural networks literature, Hitron and Parter \cite{Merav2019} tackled a similar problem. Their deterministic neural counter problem is our TSC problem. This work differ in three ways. First, our network has time bound $1$ while theirs is $ O(\log T)$. Second, we provide a time lower bound result and show our time bound is optimal. Third, they additionally consider an approximate version of the problem while we consider other forms of neural coding.

Our work follows similar approaches to Lynch et al. \cite{LMP17ITCS, LMP17DISC,LM18} by treating neurons as static circuits to explore the computational power of neural circuits. There are three noteworthy points about our model. First, instead of a stochastic model, we use a deterministic one. However, it should be noted that all the results in this work would still hold under the randomized model of Lynch et al. \cite{LMP17ITCS, LMP17DISC,LM18} with high probability. Second, we use a model that resets the potential at every round. Therefore, to retain temporal information, many self-excitation connections are employed in our networks. At the other extreme, we could have a model in which the potential does not decay from past rounds. In that model, temporal information can be stored in potentials, but it might require different mechanisms to translate the information from potentials to spikes. The two models thus could lead to different possible computational principles in brains. Third, we used a discrete time model instead of a continuous time model, which would be more biologically plausible. However, this might not be a concern since we could use Maass's synchronization module \cite{Maass1996} to simulate our discrete time model from a continuous time model.

In addition, our networks are not noise tolerant, whereas the actual neuronal dynamics are highly noisy. It will be interesting to consider a noise tolerant version of the network. One possible formulation is the following: at each time step $t$, with probability $\tau$ which does not depend on the number of neurons, a spiking event becomes a non-spike event. Can the network still count exactly or approximately with high probability? Can we find a noise tolerant network that can do this with $O(\log T)$ neurons?

Another aspect of the temporal input we have not explored is the time-scale invariance of the problem. In biology, many problems are time-scale invariant. A person who says ``apple" fast can be understood as well as a person who says ``apple" slowly. If we exploit this invariance, we might be able to reduce the networks' complexity further.

\bibliography{main}
\bibliographystyle{alpha}

\end{document}